\documentclass{article}

\title{Complementarity in categorical\\quantum mechanics}
\author{Chris Heunen\thanks{Oxford University Computing Laboratory, supported by the Netherlands Organisation for Scientific Research (NWO).}}

\usepackage{amssymb,amsmath,xspace}
\usepackage[curve]{xypic}
\usepackage{tikz}
\usepackage[amsmath,thmmarks]{ntheorem}

\usetikzlibrary{arrows,positioning,matrix,backgrounds}
\tikzset{bend angle=45, 
         baseline=(current bounding box.center), 
         inner sep=0ex,
         dot/.style={circle, draw=black, fill=black!50, inner sep=.33ex},
         box/.style={inner sep=.3ex, draw, rounded corners=.3ex},
         column sep=2ex,
         row sep=2ex}

\newcommand{\after}{\circ}

\newcommand{\cat}[1]{\ensuremath{\mathbf{#1}}}
\newcommand{\Cat}[1]{\ensuremath{\mathbf{#1}}}

\newcommand{\idmap}[1][]{\ensuremath{\mathrm{id}_{#1}}}
\newcommand{\id}[1][]{\idmap[#1]}

\newcommand{\op}{\ensuremath{^{\mathrm{op}}}}

\newcommand{\field}[1]{\ensuremath{\mathbb{#1}}}

\newcommand{\inprod}[2]{\ensuremath{\langle #1\,|\,#2 \rangle}}

\newcommand{\powerset}{\ensuremath{\mathcal{P}}}

\newcommand{\tensor}{\ensuremath{\otimes}}

\newcommand{\ie}{\textit{i.e.}~}
\newcommand{\eg}{\textit{e.g.}~}

\newcommand{\coker}{\ensuremath{\mathrm{coker}}}

\newcommand{\KSub}{\ensuremath{\mathrm{KSub}}}
\newcommand{\Proj}{\ensuremath{\mathrm{Proj}}}

\newcommand{\cC}{\ensuremath{\mathcal{C}}}

\renewcommand{\implies}{\ensuremath{\Rightarrow}}

\newcommand{\qedhere}{\ensuremath{\tag*{$\Box$}}}

\newcommand{\pullback}[1][dr]{\save*!/#1-1.2pc/#1:(-1,1)@^{|-}\restore}
\makeatother
 \newdir{ >}{{}*!/-7.5pt/@{>}}
 \newdir{|>}{!/4.5pt/@{|}*:(1,-.2)@^{>}*:(1,+.2)@_{>}}
 \newdir{ |>}{{}*!/-3pt/@{|}*!/-7.5pt/:(1,-.2)@^{>}*!/-7.5pt/:(1,+.2)@_{>}}
 \newdir{ C}{{}*!/-7.5pt/@^{(}}
 \newdir{ c}{{}*!/-7.5pt/@_{(}}

\newcommand{\zeromono}{\ar@{ >->}|-*@{o}}
\newcommand{\zeroepi}{\ar@{->>}|-*@{o}}
\newcommand{\zeromonoepi}{\ar@{ >->>}|-*@{o}}

\newcommand{\xyline}[2][]{\ensuremath{\smash{\xymatrix@1#1{#2}}}}
\newcommand{\xyinline}[2][]{\ensuremath{\smash{\xymatrix@1#1{#2}}}^{\rule[8.5pt]{0pt}{0pt}}}

\theoremnumbering{arabic}
\theoremstyle{plain}
\theorembodyfont{\itshape}
\theoremheaderfont{\normalfont\bfseries}
\theoremseparator{}
\newtheorem{theorem}{Theorem}
\newtheorem{lemma}[theorem]{Lemma}
\newtheorem{proposition}[theorem]{Proposition}
\newtheorem{corollary}[theorem]{Corollary}

\theorembodyfont{\normalfont}
\newtheorem{definition}[theorem]{Definition}
\newtheorem{example}[theorem]{Example}

\theoremstyle{nonumberplain}
\theoremheaderfont{\scshape}

\newtheorem{proof}{Proof}

\qedsymbol{\ensuremath{\Box}}

\begin{document}

\maketitle

\begin{abstract}
  We relate notions of complementarity in three layers of
  quantum mechanics: (i) von Neumann algebras, (ii) Hilbert spaces,
  and (iii) orthomodular lattices. Taking a more general categorical
  perspective of which the above are instances, we consider dagger
  monoidal kernel categories for (ii), so that (i) become
  (sub)endohomsets and (iii) become subobject lattices.  
  By developing a `point-free' definition of copyability
  we link (i) commutative von Neumann subalgebras, (ii) classical
  structures, and (iii) Boolean subalgebras.
\end{abstract}

\section{Introduction}

Complementarity is a supporting pillar of the Copenhagen
interpretation of quantum mechanics. Unfortunately, Bohr's own
formulation of the principle remained imprecise and
flexible~\cite{scheibe:quantumlogic}, and to date there is no concensus
on a clear mathematical definition.
Here, we understand it, roughly, to mean that complete knowledge of a 
quantum system can only be attained through examining all of its possible
classical subsystems~\cite{held:complementarity}. Notice that, perhaps
unlike Bohr's own, this interpretation concerns \emph{all} classical
contexts, leading to a weaker notion of binary complementarity than
usual. To avoid clashes with the various existing terminologies and
their connotations, and to emphasize the distinction between talking
about \emph{two} (totally) incompatible classical contexts (as Bohr
typically did), and mentioning all of them, we will speak of
\emph{partially complementary} classical contexts only when 
considering \emph{two} of them. Only taken all together, 
(pairwise partially complementary) classical contexts give complete
information, and we call them \emph{completely complementary}.
This paper considers instances of this interpretation of
complementarity with regard to three aspects of quantum mechanics.    
\begin{enumerate}
  \item[(i)] The observables of a quantum system form a von Neumann
    algebra. In this setting, complete complementarity is customarily
    taken to mean that one has to look at all commutative von Neumann
    subalgebras~\cite{landsman:between,strocchi:infinite,butterfieldisham:kochenspecker1}.  
  \item[(ii)] The states of a quantum system are unit vectors in a Hilbert
    space, which can be coordinatized by choosing any orthonormal
    basis. Here, complete complementarity may be interpreted as saying
    that it takes measurements in all possible orthonormal bases (of
    many identical copies of a system) to determine its state
    perfectly~\cite{vonneumann:grundlagen}, as in state
    tomography~\cite{nielsenchuang:quantumcomputation,parthasarathy:estimation}.    
 \item[(iii)] The measurable properties of a quantum system form an
    orthomodular lattice. Complete complementarity translates to this
    view as stating that the lattice structure is determined by
    all Boolean
    sublattices~\cite{piron:foundations,kalmbach:orthomodularlattices}.  
\end{enumerate}

In fact, we will take a more general perspective, as all three layers,
separately, have recently been studied categorically.
\begin{enumerate}
  \item[(i)] The set of commutative von Neumann subalgebras of a von
    Neumann algebra gives rise to a topos of set-valued functors, whose
    intuitionistic internal logic sheds light on the the original
    noncommutative algebra in so far as complete complementarity is
    concerned~\cite{heunenlandsmanspitters:bohrification,doeringisham:thing}.  
  \item[(ii)] The category of Hilbert spaces can be abstracted to a dagger
    monoidal category, in which much of quantum mechanics can still
    be formulated~\cite{abramskycoecke:categoricalquantum}. In this
    framework, orthonormal bases are characterized as so-called
    classical
    structures~\cite{coeckepavlovic:classicalobjects,coeckepavlovicvicary:bases,abramskyheunen:hstar}. 
  \item[(iii)] Orthomodular lattices can be obtained as kernel
    subobjects in a so-called dagger kernel
    category~\cite{heunenjacobs:daggerkernellogic}. This paper
    considers Boolean sublattices systematically, in the tradition of
    \eg\cite{kochenspecker:hiddenvariables}.
\end{enumerate}

We will take the view that of these three layers, (ii) is the
primitive one, which the others derive from. Indeed, our main results
are in categories that are simultaneously dagger monoidal categories
and dagger kernel categories. We give definitions of partial and
complete complementarity for (i) commutative von Neumann subalgebras,
(ii) classical structures, and (iii) Boolean sublattices of the
orthomodular lattice of kernels. By developing a point-free notion of
copyability, we obtain a bijective correspondence between
partially complementary classical structures and partially
complementary Boolean sublattices. It is worth mentioning that this
seems to be the first positive use of tensor products in the study of
orthomodular lattices---the relation between tensor products and
orthomodular lattices has resisted attempts at structural
understanding so far, and only negative, restrictive, results are known.
Our second main contribution is to characterize 
categorically what partially complementary commutative von Neumann
subalgebras correspond to in terms of classical structures in the
category of Hilbert spaces, conceptually improving upon previous work
in this setting~\cite{petz:complementarity,parthasarathy:estimation}.
The plan of the paper is as follows:
Sections \ref{sec:classicalstructures}, \ref{sec:orthomodular} and
\ref{sec:vonneumannalgebras} study layers (ii), (iii) and (i)
respectively. Conclusions are then drawn in Section
\ref{sec:conclusion}. 
The author is grateful to Samson Abramsky, Ross Duncan, Klaas
Landsman, and Jamie Vicary for useful pointers and discussions.

\section{Classical structures}
\label{sec:classicalstructures}

A \emph{dagger} on a category $\cat{D}$ is a functor $\dag
\colon \cat{D}\op \to \cat{D}$ that acts on objects as $X^\dag=X$ and
satisfies $f^{\dag\dag}=f$ on morphisms. We will be interested in
dagger categories that also have tensor products and
kernels. By way of introduction we first recall these two extra structures
separately, and then consider how they cooperate. 

A \emph{dagger symmetric monoidal category} is a dagger category that
is simultaneously symmetric monoidal, satisfies
$(f \tensor g)^\dag = f^\dag \tensor g^\dag$, and whose coherence
isomorphisms such as $\lambda \colon X \tensor I \to X$ satisfy
$\lambda^{-1}=\lambda^\dag$. For more information about dagger
monoidal categories and their uses in physics, we refer
to~\cite{coeckepaquettepavlovic:structuralism,abramskycoecke:categoricalquantum}. A morphism $f$ is called \emph{dagger monic} when $f^\dag \after f = \id$.

\begin{definition}
  A \emph{classical structure} in a dagger symmetric
  monoidal category $\cat{D}$ is a commutative semigroup $\delta
  \colon X \to X \tensor X$ that satisfies $\delta^\dag \after \delta
  = \id$ and the following so-called H*-axiom: there is an involution
  $* \colon \cat{D}(I,X)\op \to \cat{D}(I,X)$ such that $\delta^\dag
  \after (x^* \tensor \id) = (x^\dag \tensor \id) \after \delta$.
\end{definition}

Spelling out this terse definition in the graphical
calculus~\cite{selinger:graphicallanguages}, its conditions 
look as follows, where $\delta$ is depicted as 
$\vcenter
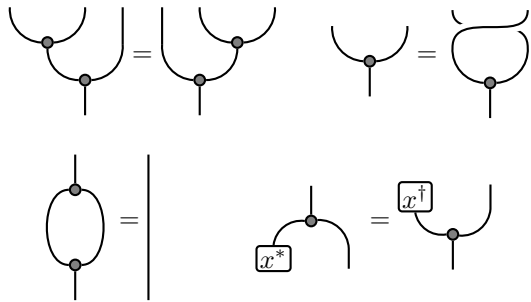
{\hbox{\begin{tikzpicture}[thick,xscale=0.2,yscale=-0.2] 
    \node (3) at (-1,-2) {};
    \node [style=dot] (1) at (0,-1) {};
    \node (0) at (0,0) {};
    \node (2) at (1,-2) {};
    \draw [bend left=45]  (1) to (2);
    \draw [bend left=45]  (3) to (1);
    \draw  (1) to (0);
  \end{tikzpicture}}}$.

\[
  \begin{tikzpicture}[thick,scale=0.5]
    \node [style=dot] (1) at (0,-1) {};
    \node (5) at (0,1) {};
    \node (0) at (0,-2) {};
    \node (6) at (1,1) {};
    \node (2) at (1,0) {};
    \node (4) at (-2,1) {};
    \node [style=dot] (3) at (-1,0) {};
    \draw (1) to (0);
    \draw [bend right=45]  (1) to (2);
    \draw [bend left=45]  (5) to (3);
    \draw (6) to (2.south);
    \draw [bend right=45]  (4) to (3);
    \draw [bend right=45]  (3) to (1);
  \end{tikzpicture}
  =
  \begin{tikzpicture}[thick,scale=0.5]
    \node [style=dot] (1) at (0,-1) {};
    \node (5) at (0,1) {};
    \node (0) at (0,-2) {};
    \node (6) at (-1,1) {};
    \node (2) at (-1,0) {};
    \node (4) at (2,1) {};
    \node [style=dot] (3) at (1,0) {};
    \draw  (1) to (0);
    \draw [bend left=45]  (1) to (2);
    \draw [bend right=45]  (5) to (3);
    \draw  (6) to (2.south);
    \draw [bend left=45]  (4) to (3);
    \draw [bend left=45]  (3) to (1);
  \end{tikzpicture}
  \qquad
  \begin{tikzpicture}[thick,scale=0.5]
    \node (3) at (-1,2) {};
    \node [style=dot] (1) at (0,1) {};
    \node (0) at (0,0) {};
    \node (2) at (1,2) {};
    \draw [bend right=45]  (1) to (2);
    \draw [bend right=45]  (3) to (1);
    \draw  (1) to (0);
  \end{tikzpicture}
  =
  \begin{tikzpicture}[thick,scale=0.5,cross/.style={preaction={draw=white, -, line width=5pt}}]
    \node [style=dot] (1) at (0,1) {};
    \node (0) at (0,0) {};
    \node (2) at (1,2) {};
    \node (3) at (-1,2) {};
    \node (4) at (1,3) {};
    \node (5) at (-1,3) {};
    \draw (1) to (0);
    \draw (1) to [bend right] (2) .. controls (4) and (3) .. (5);
    \draw[cross] (1) to [bend left] (3) .. controls (5) and (2) .. (4);
  \end{tikzpicture}
  \]
  \[
  \begin{tikzpicture}[thick,scale=0.5]
    \node (0) at (0,-2) {};
    \node (3) at (0,2) {};
    \node[dot] (2) at (0,1) {};
    \node[dot] (1) at (0,-1) {};
    \draw [bend right=90]  (2) to (1);
    \draw  (2) to (3);
    \draw [bend right=90]  (1) to (2);
    \draw  (0) to (1);
  \end{tikzpicture}
  = 
  \begin{tikzpicture}[thick,scale=0.5]
    \node (0) at (0,-2) {};
    \node (3) at (0,2) {};
    \draw  (0) to (3);
  \end{tikzpicture}
  \qquad\qquad
  \begin{tikzpicture}[thick,scale=0.5]
    \node [style=dot] (1) at (0,0) {};
    \node [box] (2) at (-1,-1) {$x^*$};
    \node (0) at (0,1) {};
    \node (3) at (1,-1) {\phantom{$x^*$}};
    \draw (0) to (1);
    \draw [bend left=45]  (1) to (3.north) -- (3.south);
    \draw [bend right=45]  (1) to (2);
  \end{tikzpicture}
  =
  \begin{tikzpicture}[thick,scale=0.5]
    \node [style=dot] (1) at (0,0) {};
    \node [box] (2) at (-1,1) {$x^\dag$};
    \node (0) at (0,-1) {};
    \node (3) at (1,1) {\phantom{$x^\dag$}};
    \draw (0) to (1);
    \draw [bend right=45]  (1) to (3.south) -- (3.north);
    \draw [bend left=45]  (1) to (2);
  \end{tikzpicture}
\]
We will not
explicitly use much of a classical structure except its type and the
fact that it is dagger monic; for more information we refer
to the forthcoming article~\cite{abramskyheunen:hstar}, and the
aforementioned~\cite{coeckepaquettepavlovic:structuralism}. 

A \emph{dagger kernel category} is a dagger category that has a zero
object $0$, and in which every morphism has a kernel that is dagger
monic. We write $\ker(f)$ for the kernel of $f$, and
$\coker(f)=\ker(f^\dag)^\dag$ for its cokernel. The definition
$k^\perp = \ker(k^\dag)$ for kernels $k$ yields an orthocomplement on
the partially ordered set $\KSub(X)$ of kernel subobjects of a fixed
object $X$. The main result of~\cite{heunenjacobs:daggerkernellogic},
which we refer to for more information about dagger kernel categories,
is that this poset $\KSub(X)$ is always an orthomodular lattice.

The goal of this section is to investigate when kernels $k \colon K \to
X$ are `compatible' with a given classical structure. To do so, we
develop a notion of copyability that has to be `point-free' because
$K$ is typically not the monoidal unit $I$.

\subsection{Kernels and tensor products}
\label{subsec:kerandtensor}

Fix a category $\cat{D}$, and assume it to be a dagger symmetric
monoidal category and a dagger kernel category simultaneously, which
additionally satisfies 
\[
  \ker(f) \tensor \ker(g) = \ker(f \tensor \id) \wedge \ker(\id \tensor g)
\]
for all morphisms $f$ and $g$.\footnote{One might consider additional
coherence requirements such as $\ker(f \tensor g)=\ker(f \tensor
\id) \vee \ker(\id \tensor g)$, but these are not necessary for our
present purposes.}
The categories $\Cat{Hilb}$ and $\Cat{Rel}$ both satisfy the above
relationship between tensor products and kernels. Some coherence
properties follow easily from the assumptions:
\begin{align*}
  \ker(f) \tensor 0 & = 0, &
  0 \tensor \ker(g) & = 0, \\
  \ker(f) \tensor \id & = \ker(f \tensor \id), &
  \id \tensor \ker(g) & = \ker(\id \tensor g), \\
  \ker(f) \tensor \id = 0 & \Leftrightarrow \ker(f) = 0, &
  \id \tensor \ker(g) = 0 & \Leftrightarrow \ker(g) = 0.
\end{align*}

Notice that requiring $\ker(f \tensor g) = \ker(f) \tensor \ker(g)$
would have been too strong, for then 
$\ker(f) \tensor \id = \ker(f) \tensor \ker(0) = \ker(f \tensor 0) =
\ker(0) = \id$ for any $f$.

\subsection{Copyability}

Throughout this section we fix a classical structure $\delta
\colon X \to X \tensor X$. 

\begin{definition}
\label{def:copyable}
  An endomorphism $p \colon X \to X$ is called \emph{copyable} (along
  $\delta$) when  
  $
    \delta \after p = (p \tensor p) \after \delta.
  $
  A nonendomorphism $k \colon K \to X$ is called \emph{copyable}
  (along $\delta$) when $P(k)=k \after k^\dag$ is. 
\end{definition}

We start by relating the previous definition to copyability
of vectors as used in~\cite{coeckeduncan:observables}.

\begin{lemma}
  The following are equivalent for a unit vector $x$ in $H \in \Cat{Hilb}$:
  \begin{enumerate}
   \item[(a)] the morphism $\field{C} \to H$ defined by $1 \mapsto x$ is copyable;
   \item[(b)] there is a phase $z \in \field{C}$ with $|z|=1$ such
     that $\delta(x) = z \cdot (x \tensor x)$; 
   \item[(c)] there is a unit vector $x' \in H$ with $P(x)=P(x')$ and
     $\delta(x')=x' \tensor x'$.
  \end{enumerate}
\end{lemma}
\begin{proof}
  For (a)$\Rightarrow$(b):
  \begin{align*}
    \delta(x) = (\delta \after P(x))(x) = (P(x) \tensor P(x)) \after
    \delta(x) = \inprod{x \tensor x}{\delta(x)} \cdot (x \tensor x).
  \end{align*}
  Taking $z = \inprod{x \tensor x}{\delta(x)}$ gives $|z| =
  \|\inprod{x \tensor x}{\delta(x)} \cdot (x \tensor x)\| =
  \|\delta(x)\| = \|x\| = 1$.

  Conversely, to see (b)$\Rightarrow$(a):
  \begin{align*}
        (\delta \after P(x))(y) 
    & = \delta \after x \after x^\dag (y) \\
    & = \inprod{x}{y} \cdot \delta(x) \\
    & = \inprod{x}{y} \cdot z \cdot (x \tensor x) \\
    & = \inprod{\delta(x)}{\delta(y)} \cdot z \cdot (x \tensor x) \\
    & = |z|^2 \cdot \inprod{x \tensor x}{\delta(y)} \cdot (x \tensor
    x) \\
    & = \inprod{x \tensor x}{\delta(y)} \cdot (x \tensor x) \\
    & = (x \tensor x) \after (x^\dag \tensor x^\dag) \after \delta(y)
    \\
    & = P(x \tensor x) \after \delta(y).
  \end{align*}
  The equivalence of (b) and (c) is established by the equality $x' = z  \cdot x$.
  \qed
\end{proof}

\begin{example}
\label{ex:zeroone}
  In any dagger kernel category with tensor products satisfying the
  coherence set out in section~\ref{subsec:kerandtensor}, zero
  morphisms and identity morphisms are always copyable:
  \begin{align*}
    \delta \after P(0) = \delta \after 0 = 0 \after \delta = (0 \tensor
    0) \after \delta = P(0 \tensor 0) \after \delta, \\
    \delta \after P(\id) = \delta \after \id = \delta = (\id \tensor
    \id) \after \delta = P(\id \tensor \id) \after \delta.
  \end{align*}
  These two kernels are called the \emph{trivial} kernels.
\end{example}

\begin{example}
\label{ex:Hilb}
  In the category $\Cat{Hilb}$ of Hilbert spaces, a classical
  structure $\delta$ corresponds to the choice of an orthonormal basis
  $(e_i)$ \cite{abramskyheunen:hstar}, whereas a kernel corresponds to a
  (closed) linear subspace \cite{heunenjacobs:daggerkernellogic}. A
  kernel is copyable if and only if it is the linear span of a subset
  of the orthonormal basis. 
\end{example}

\begin{example}
\label{ex:Rel}
  In the category $\Cat{Rel}$ of sets and relations, a classical
  structure $\delta$ on $X$ corresponds to (a disjoint union of) Abelian
  group structure(s) on $X$ \cite{pavlovic:frobeniusinrel}, and a kernel
  is to a subset $K \subseteq X$
  \cite{heunenjacobs:daggerkernellogic}. A kernel is copyable iff
  \begin{align*}
        \{(x\cdot y,(x,y)) \mid x,y \in X, x \cdot y \in K \}
    & = \delta \after P(k) \\
    & = P(k \tensor k) \after \delta \\
    & = \{(x\cdot y,(x,y)) \mid x \in K, y \in K \},
  \end{align*}
  \ie if and only if $x \in K \wedge y \in K \Leftrightarrow x \cdot y
  \in K$. One direction of this equivalence implies that $K$ is a
  subsemigroup. Fixing $k \in K$, we see that for any $x \in X$ there
  is $y=x^{-1} \cdot k$ such that $x \cdot y \in K$. Therefore, the
  other direction implies that $x \in K$. That is, $K=X$.
  We conclude that the only copyable kernels in $\Cat{Rel}$ are the
  trivial ones.

  This signifies that $\Cat{Rel}$ does not `have enough kernels'. 
  There is an order isomorphism between $\{ p \in \Cat{Rel}(X,X) \mid
  p=p^\dag=p \after p \leq \id\}$ and $\KSub(X)$ for a certain order $\leq$ on
  endohomsets~\cite[Proposition~12]{heunenjacobs:daggerkernellogic}.
  The situation improves when we consider all $p \in \Cat{Rel}(X,X)$
  with $p=p^\dag=p^2$ instead of just the ones below the identity:
  these are precisely \emph{partial equivalence relations}, \ie
  the symmetric and transitive relations. One finds that such
  an endomorphism $\sim$ is copyable if and only if it is a `groupoid
  congruence' in the following sense:
  \[
    x\cdot y \sim z \Longleftrightarrow \exists_{x',y'}[x \sim x', y
    \sim y', x'\cdot y'=z]
  \]
\end{example}

\begin{lemma}
\label{lem:copyablefillin}
  A dagger monic $k$ is copyable if and only if there is a (unique)
  morphism $\delta_k$ making the following diagram commute:
  \[\xymatrix{
      X \ar@{ >->}_-{\delta}[d] \ar@{->>}^-{k^\dag}[r] 
    & K \ar@{-->}^-{\delta_k}[d] \ar@{ >->}^-{k}[r] 
    & X \ar@{ >->}^-{\delta}[d] \\
      X \tensor X \ar@{->>}_-{k^\dag \tensor k^\dag}[r] 
    & K \tensor K \ar@{ >->}_-{k \tensor k}[r] 
    & X \tensor X.
  }\]
\end{lemma}
\begin{proof}  
  \begin{align*}
    &                     k \text{ is copyable} \\
    & \Longleftrightarrow 
      (P(k) \tensor P(k)) \after \delta = \delta \after P(k) \\
    & \Longleftrightarrow 
      (P(k) \tensor P(k)) \after \delta = (P(k) \tensor P(k))
      \after \delta \after P(k) = \delta \after P(k) \\      
    & \Longleftrightarrow 
      \exists_{\delta_k} . \delta \after k = (k \tensor k) \after \delta_k,
      \delta_k \after k^\dag = (k^\dag \tensor k^\dag) \after \delta.
    \qedhere
  \end{align*}
\end{proof}

We say that $f$ is a \emph{dagger retract} of $g$ if there are dagger
monic $a$ and $b$ making the following diagram commute:
\[\xymatrix{
    \cdot \ar^-{a}[r] \ar_-{f}[d] 
  & \cdot \ar^-{a^\dag}[r] \ar^-{g}[d]
  & \cdot \ar^-{f}[d] \\
    \cdot \ar_-{b}[r]
  & \cdot \ar_-{b^\dag}[r] 
  & \cdot
}\]
Notice that if $f$ and $f'$ are both dagger retracts of $g$ (along
the same $a$ and $b$), then $f=b^\dag \after b \after f = b^\dag
\after g \after a = f' \after a^\dag \after a = f'$.

\begin{proposition}
\label{prop:restrictedfrob}
  If $k$ is a copyable dagger monic, $\delta_k$ is a classical structure.
\end{proposition}
\begin{proof}
  If $k$ is a copyable dagger monic, then it follows from
  Lemma~\ref{lem:copyablefillin} that $\delta_k$ is a dagger retract of
  $\delta$, and $\delta_k^\dag$ 
  is a dagger retract of $\delta^\dag$. Therefore, $\delta_k$ 
  is associative, commutative, and is dagger monic. For
  example, to verify commutativity, notice that $\gamma_k \colon K
  \tensor K \to K \tensor K$ is a dagger retract of $\gamma \colon X \tensor
  X \to X \tensor X$. Since dagger retracts compose, this means that
  $\gamma_k \after \delta_k$ and $\delta_k$ are both dagger retracts
  (along the same morphisms) of $\gamma \after \delta = \delta$. Hence
  $\gamma_k \after \delta_k = \delta_k$. The other algebraic
  properties are verified similarly (including the Frobenius equation). 
  
  We are left to check the H*-axiom. Let $x \colon I \to
  K$. Since $\delta$ satisfies the H*-axiom, there is $(k \after x)^*
  \colon I \to X$ such that 
  \[
      \delta^\dag \after ((k \after x)^* \tensor \id)
    = ((k \after x)^\dag \tensor \id) \after \delta.
  \]
  Now put $x^* = k^\dag \after (k \after x)^* \colon I \to K$. Then:
  \begin{align*}
        \delta_k^\dag \after (x^* \tensor \id)
    & = \delta_k^\dag \after (k^\dag \tensor \id) 
                     \after ((k \after x)^* \tensor \id)  \\
    & = \delta_k^\dag \after (k^\dag \tensor k^\dag) 
                    \after ((k \after x)^* \tensor \id) \after k \\
    & = k^\dag \after \delta^\dag \after ((k \after x)^* \tensor \id)
              \after k \\
    & = k^\dag \after ((k \after x)^\dag \tensor \id) 
              \after \delta \after k \\
    & = (x^\dag \tensor \id) \after (k^\dag \tensor k^\dag) 
                            \after \delta \after k \\
    & = (x^\dag \tensor \id) \after \delta_k \after k^\dag \after k \\
    & = (x^\dag \tensor \id) \after \delta_k.
  \end{align*}
  Hence $\delta_k$ satisfies the H*-axiom, too.
  \qed  
\end{proof}

Observe that it follows from the previous proposition that a dagger
monic $k$ is copyable if and only if its domain carries a classical
structure $\delta_k$ and $k$ is simultaneously a homomorphism of
nonunital monoids and of nonunital comonoids.

\subsection{Complementarity and mutual unbiasedness}

\begin{definition}
\label{def:complementaryclassicalstructures}
  Two classical structures are \emph{partially complementary}
  if no nontrivial kernel is simultaneously copyable along both.
\end{definition}

The above definition clashes with complementarity of classical
structures as defined in \cite{coeckeduncan:observables}. Let us 
spend some time developing a notion that does correspond to
complementarity in the sense of \cite{coeckeduncan:observables}
directly. 

\begin{definition}
  A morphism $x \colon U \to X$ is called \emph{unbiased} (relative to
  $\delta$) if and only if $P(x^\dag \after k) = P(x^\dag \after l)$ for all
  copyable kernels $k$ and $l$.
\end{definition}

Recall that if the ambient category is simply well-pointed, then
$\KSub(X)$ is atomic, and its atoms are precisely the kernels with domain
$I$~\cite{heunenjacobs:daggerkernellogic}. If such `points' are
unbiased in the above sense, then they are unbiased vectors in the sense
of~\cite{coeckeduncan:observables}. 

Another advantage of the previous definition in this point-based setting
is that it does not need to specify what the scalars $\inprod{k}{x}$ 
are. In the traditional point-based setting this scalar involves the
dimension of the carrier Hilbert space, and is therefore limited to
finite-dimensional spaces. The above definition can be interpreted
regardless of dimensional aspects.




\begin{lemma}
\label{lem:ifcopyablethennotunbiased}
  If a nonzero kernel is copyable (along $\delta$) then it is not
  unbiased (relative to $\delta$). 
\end{lemma}
\begin{proof}
  Let $x$ be a nonzero copyable kernel.
  Then $x^\perp$ is copyable, too, by Lemma~\ref{lem:orthogonalcopyable}.
  The trivial kernels are always copyable by Example~\ref{ex:zeroone}. 
  Hence $k=\id$ and $l=x^\perp$ are both copyable kernels. But
  \[
      P(x^\dag \after k)
    = P(x^\dag) 
    = x^\dag \after x 
    \neq 0 
    = P(0)
    = P(x^\dag \after \ker(x^\dag)) 
    = P(x^\dag \after l),
  \]
  and therefore $x$ cannot be unbiased.
  \qed
\end{proof}

\begin{definition}
\label{def:mutuallyunbiased}
  Two classical structures are \emph{mutually unbiased}
  if a nontrivial kernel is unbiased relative to one whenever it is copyable
  along the other. 
\end{definition}

\begin{proposition}
  Mutually unbiased classical structures are partially complementary. 
\end{proposition}
\begin{proof}
  Suppose that $\delta$ and $\delta'$ are mutually unbiased,
  and let $k$ be a kernel in the intersection. That is,
  $k$ is copyable along both $\delta$ and $\delta'$. By
  Lemma~\ref{lem:ifcopyablethennotunbiased}, $k$ cannot be unbiased
  relative to $\delta'$. This contradicts mutual unbiasedness unless
  $k$ were a trivial kernel.
  \qed  
\end{proof}

There is no converse to the previous proposition. For example, consider
the object $\field{C}^2$ in the category $\Cat{Hilb}$. A classical
structure corresponds to an orthonormal basis of $\field{C}^2$, and
hence corresponds (up to sign) to a single ray. The collections of copyables of
classical structures induced by two different rays always have trivial
intersection. But certainly not every pair of different orthonormal
bases is mutually unbiased.  We conclude that
Definition~\ref{def:mutuallyunbiased} is too strong for our purposes.

\section{Boolean subalgebras of orthomodular lattices}
\label{sec:orthomodular}

This section concerns level (iii) of the Introduction. We will prove
that kernels that are copyable along $\delta$ form a sublattice of the orthomodular lattice of all kernel subobjects of $X$, that is be Boolean as soon as it is closed under complements.

\begin{lemma}
\label{lem:meetcopyable}
  The copyable kernels form a sub-meetsemilattice of $\KSub(X)$.
\end{lemma}
\begin{proof}
  The bottom element 0 is always copyable by Example~\ref{ex:zeroone}. So we
  have to prove that if $k$ and $l$ are copyable 
  kernels, then so is $k \wedge l$. Recall that $k \wedge l$ is
  defined as the pullback.
  \[\xymatrix{
    K \wedge L \pullback \ar^-{p}[r] \ar_-{q}[d] & L \ar^-{l}[d] \\
    K \ar_-{k}[r] & X
  }\]
  Together with the assumption that $k$ and $l$ are copyable, this
  means that the top, back, right and bottom face of the following
  cube commute:
  \[\xymatrix@C-3ex@R-3ex{
    & K \ar^-{k}[rr] & & X \\
    K \wedge L \ar^-{q}[ur] \ar^(.75){p}[rr] & & L \ar^-{l}[ur] \\
    & K \tensor K \ar|-{\hole}^(.25){\delta_k^\dag}[uu] 
      \ar|-{\hole}_(.25){k \tensor k}[rr]
    & & X \tensor X. \ar_-{\delta^\dag}[uu]\\
    (K \wedge L)^{\tensor 2} \ar@{-->}^-{\varphi}[uu] \ar^-{q
      \tensor q}[ur] \ar_-{p \tensor p}[rr] 
    & & L \tensor L \ar_-{l \tensor l}[ur] \ar_(.75){\delta_l^\dag}[uu]
  }\]
  Hence $l \after \delta_l^\dag \after (p \tensor p) = k \after
  \delta_k^\dag \after (q \tensor q)$. Therefore, by the universal
  property of pullbacks, there exists a dashed morphism $\varphi$
  making the left and front sides of the above cube commute.
  Using the fact that $p$ and $q$ are dagger monic, we deduce $\varphi
  = (k \wedge l)^\dag \after \delta^\dag \after ((k \wedge l) \tensor
  (k \wedge l))$. 
  This means that the left square in the following diagram commutes:
  \[\xymatrix@C+6ex{
      X \ar^-{(k \wedge l)^\dag}[r] \ar_-{\delta}[d] 
    & K \wedge L \ar^-{\varphi^\dag}[d] \ar^-{k \wedge l}[r]
    & X \ar^-{\delta}[d] \\
      X \tensor X \ar_-{(k \wedge l)^\dag \tensor (k \wedge l)^\dag}[r] 
    & (K \wedge L) \tensor (K \wedge L) \ar_-{(k \wedge l) \tensor (k
      \wedge l)}[r]
    & X \tensor X.
  }\]
  The right square is seen to commute analogously---take daggers of all
  the vertical morphisms in the cube. Therefore the whole rectangle 
  commutes. In other words, $\delta \after P(k \wedge l) = (P(k \wedge
  l) \tensor P(k \wedge l)) \after \delta$, that is, $k \wedge l$ is
  copyable. 
  \qed
\end{proof}

\begin{lemma}
\label{lem:Booleanness}
  \cite[Theorem~1]{heunenjacobs:daggerkernellogic}
  An orthocomplemented sublattice $L$ of $\KSub(X)$ is Boolean if
  and only if $k \wedge l = 0$ implies $l^\dag \after k = 0$ for all
  $k,l \in L$.
  \qed
\end{lemma}

\begin{theorem}
\label{thm:copyableboolean}
  The sub-meetsemilattice of copyable kernels of the orthomodular lattice $\KSub(X)$ satisfies $k \wedge l \implies l^\dag \circ k = 0$.
\end{theorem}
\begin{proof}
  Let $k$ and $l$ be copyable 
  kernels and suppose $k \wedge l = 0$. Say $k=\ker(f)$ and
  $l=\ker(g)$. Then
  \[
    (f \tensor \id) \after (k \tensor l) = (f \after k) \tensor l = 0
    \tensor l = 0,
  \]
  so that $k \tensor l \leq \ker(f \tensor \id) = k \tensor \id \leq
  (k \tensor \id) \wedge (\id \tensor k) = k \tensor k$. Similarly, $k
  \tensor l \leq l \tensor l$. Therefore the bottom, top, back and right
  faces of the following cube commute:
  \[\xymatrix@C-1ex@R-2ex{
    & K \ar^-{k}[rr] & & X \\
    0 \ar^-{0}[ur] \ar^(.75){0}[rr] & & L \ar^-{l}[ur] \\
    & K \tensor K \ar|-{\hole}^(.25){\delta_k^\dag}[uu] 
      \ar|-{\hole}^(.25){k \tensor k}[rr]
    & & X \tensor X. \ar_-{\delta^\dag}[uu]\\
    K \tensor L \ar@{-->}^-{\varphi}[uu] \ar|-{\id \tensor (k^\dag
      \after l)}[ur] \ar_-{(l^\dag \after k) \tensor \id}[rr] \ar|-{k
      \tensor l}|(.67){\hole}[rrru]
    & & L \tensor L \ar_-{l \tensor l}[ur] \ar_(.75){\delta_l^\dag}[uu]
  }\]
  The universal property of the pullback formed by the top face yields
  the dashed morphism $\varphi$ making the left and front faces
  commute. Hence $\delta_l^\dag \after ((l^\dag \after k) \tensor \id)
  = 0$. But then, as $k$ and $l$ are copyable:
  \begin{align*}
        l^\dag \after k
    & = l^\dag \after k \after \delta_k^\dag \after \delta_k \\
    & = l^\dag \after \delta^\dag \after (k \tensor k) \after \delta_k \\
    & = \delta_l^\dag \after (l^\dag \tensor l^\dag) 
        \after (k \tensor k) \after \delta_k \\
    & = \delta_l^\dag \after ((l^\dag \after k) \tensor \id)
        \after (\id \tensor (l^\dag \after k)) \after \delta_k \\
    & = 0 \after (\id \tensor (l^\dag \after k)) \after \delta_k \\
    & = 0.
  \end{align*}
  This finishes the proof.
  \qed  
\end{proof}

Hence as soon as copyable elements of $\KSub(X)$ are closed under complements $k^\perp = \ker(k^\dag)$, they form a Boolean subalgebra. This is the case in $\cat{Hilb}$, but not in $\cat{Rel}$.
In the category $\Cat{Hilb}$, Example~\ref{ex:Hilb} shows that the copyable
kernels in fact form a maximal Boolean subalgebra. Example~\ref{ex:Rel} shows
that this does not hold generally. For in $\Cat{Rel}$, one
has $\KSub(X) = \powerset(X)$, but the copyable kernels are always
$\{0,1\}$, which are only maximal if $X$ has cardinality
$1$. Similarly, not every maximal Boolean subalgebra $B$ of $\KSub(X)$
induces a classical structure on $X$ of which $B$ are the copyables,
as does happen to be the case in $\Cat{Hilb}$. 

The following definition
expresses the standard view in (order-theoretic) quantum logic that
Boolean subalgebras of orthomodular lattices are regarded as 
embodying complete complementarity. It is precisely what is needed to
make Theorem~\ref{thm:complementarybaseslattices} true.

\begin{definition}
\label{def:complementarylattices}
  Two sub-meetsemilattices of an orthomodular lattice are called
  \emph{partially complementary} when they have trivial intersection.
\end{definition}

\begin{theorem}
\label{thm:complementarybaseslattices}
  Two classical structures are partially complementary if and only if
  their collections of copyable kernels are partially complementary.
  \qed
\end{theorem}

Hence we have linked, fully abstractly, partial complementarity in the
order-theoretic sense to partial complementarity in the sense of classical
structures. 
Various order-theoretic questions about the lattices of copyable
kernels suggest themselves for further investigation. For example, one
could imagine that the width or height of the Boolean sublattice of
copyable kernels is independent of the classical structure, enabling a
general notion of dimension. One could also study how copyability
interacts with closed or compact structure in the category.

\section{Von Neumann algebras}
\label{sec:vonneumannalgebras}

Finally, this section advances to level (i) of the Introduction. We
instantiate the dagger monoidal kernel category $\cat{D}$ to be 
$\Cat{Hilb}$. For any object $H \in \Cat{Hilb}$, the endohomset
$A=\Cat{Hilb}(H,H)$ is then a type I von Neumann algebra (and every
type I von Neumann algebra is of this form). At this
level, the notion of complete complementarity is formalized by
considering all commutative von Neumann subalgebras $C$ of $A$. We
denote the collection of all such subalgebras of $A$ by $\cC(A)$. Let
us recall some facts about this situation. 
\begin{enumerate}
  \item[(a)] The set $\Proj(A) = \{ p \in A \mid p^\dag = p = p^2 \}$ of
    projections is a complete, atomic, atomistic orthomodular
    lattice~\cite[p85]{redei:quantumlogic}.
  \item[(b)] There is an order isomorphism $\Proj(A) \cong
    \KSub(H)$~\cite[Proposition~12]{heunenjacobs:daggerkernellogic}.
  \item[(c)] Any von Neumann algebra is generated by its
    projections~\cite[6.3]{redei:quantumlogic}, so in particular 
    $C = \Proj(C)''$. 
  \item[(d)] Since $C$ is a subalgebra of $A$, also $\Proj(C)$ is a
    sublattice of $\Proj(A)$.
  \item[(e)] Because $C$ is commutative, $\Proj(C)$ is a Boolean
    algebra~\cite[4.16]{redei:quantumlogic}.
\end{enumerate}
The following lemma draws a conclusion of interest from these facts.

\begin{lemma}
\label{lem:CABooleanKSub}
  Commutative von Neumann subalgebras $C$ of $A=\Cat{Hilb}(H,H)$ are
  in bijective correspondence with Boolean subalgebras of $\KSub(H)$. 
\end{lemma}
\begin{proof}
  As in (b) above, we identify $\Proj(A)$ with $\KSub(H)$.
  A commutative subalgebra $C$ corresponds to the Boolean subalgebra
  $\Proj(C)$. Conversely, a Boolean subalgebra $B$ corresponds to the
  commutative subalgebra $B''$ generated by it. These mappings are
  inverses because $\Proj(C)'' = C$ and $\Proj(B'')=B$.
  \qed
\end{proof}

We now set out to establish the relation between commutative
subalgebras of $A$ and classical structures on $H$.

\begin{lemma}
\label{lem:BooleanKSubcopyables}
  An orthocomplemented sublattice $L$ of $\KSub(H)$ is Boolean if and
  only if the following equivalent conditions hold:
  \begin{itemize}
    \item there exists a classical structure on the greatest element
      of $L$ along which every element of $L$ is copyable;
    \item there exists a classical structure on $H$ along which every
      element of $L$ is copyable.  
  \end{itemize}
\end{lemma}
\begin{proof}
  Necessity is established by Theorem~\ref{thm:copyableboolean}.
  For sufficiency, let $L$ be a Boolean sublattice of $\KSub(H)$.
  Since $\KSub(H)$ is complete by (a) above, $\bigvee L$ exists. By
  atomicity (a), $\bigvee L$ is completely determined by the set of
  atoms $a_i$ below it. By definition of atoms, $a_i \wedge a_j = 0$
  when $i \neq j$. Because $L$ is Boolean, it follows from
  Lemma~\ref{lem:Booleanness} that $a_i$ and
  $a_j$ are orthogonal. Also, because $\Cat{Hilb}$ is simply
  well-pointed, the kernels $a_i$ correspond to one-dimensional
  subspaces~\cite[Lemma~11]{heunenjacobs:daggerkernellogic}. That is,
  the $a_i$ give an orthonormal basis for (the domain of) the greatest
  element of $L$ (which can be extended to an orthonormal basis of
  $H$). This, in turn, induces a classical structure 
  $\delta$ on the greatest element of $L$ (or
  $H$)~\cite{abramskyheunen:hstar}. Finally, Example~\ref{ex:Hilb}
  shows that the kernels $a_i$, and hence all $l \in L$, are copyable
  along $\delta$. 
  \qed  
\end{proof}

\begin{theorem}
\label{thm:CA}
  For the von Neumann algebra $A=\Cat{Hilb}(H,H)$:
  \begin{align*}
    \cC(A) \cong 
    \{ L \subseteq \KSub(H) \mid\; & L \mbox{ orthocomplemented
      sublattice}, \\ & \exists_{\delta \colon 1_L \to 1_L \tensor
      1_L}\forall_{l \in L} \,.\, l \mbox{ copyable along }\delta \}.
  \end{align*}
\end{theorem}
\begin{proof}
  This is just a combination of Lemma~\ref{lem:CABooleanKSub} and
  Lemma~\ref{lem:BooleanKSubcopyables}.
  \qed
\end{proof}

The previous theorem implies that for any classical structure $\delta$
on $H$, there is an induced commutative von Neumann subalgebra $C \in \cC(A)$
corresponding to the lattice $L$ of all copyable kernels. The following
definition and corollary finish the connections of partial complementarity
across the three levels discussed in the Introduction. 

\begin{definition}
  Von Neumann subalgebras of $\Cat{Hilb}(H,H)$ are \emph{partially
  complementary} when their intersection is the trivial subalgebra
  $\{z \cdot \id \mid z \in \field{C}\}$. 
\end{definition}

Notice that this definition does not need the subalgebras to be
commutative. It has no need for finite dimension, either, in contrast
to works that rely on the Hilbert-Schmidt inner product to make
$\Cat{Hilb}(H,H)$ into a Hilbert space~\cite{petz:complementarity}. 
Compare also~\cite[Definition~1.1]{parthasarathy:estimation}.

\begin{corollary}
  Two classical structures on an object $H$ in $\Cat{Hilb}$ are
  partially complementary if and only if they induce partially
  complementary commutative von Neumann subalgebras of $\Cat{Hilb}(H,H)$.  
  \qed
\end{corollary}

Unlike $\Proj(A)$, the sublattice $\Proj(C)$ is not atomic for general
$C \in \cC(A)$; for a counterexample, take $H=L^2([0,1])$ and
$C=L^\infty([0,1])$. If this does happen to be the case, for example if we
restrict the ambient category $\cat{D}$ to that of finite-dimensional
Hilbert spaces, we can strengthen the characterization of $\cC(A)$ in
Theorem~\ref{thm:CA}. 

\begin{proposition}
\label{prop:fdCAker}
  For a finite-dimensional Hilbert space $H$ and the von Neumann
  algebra $A=\Cat{fdHilb}(H,H)$:
  \begin{align*}
    \cC(A) \cong \{ (k_i)_{i \in I} \mid \; & 
      k_i \in \KSub(H) \backslash \{0\}, \;
      k_i \wedge k_j = 0 \mbox{ for }i \neq j,  \\ 
      & \exists_{\delta \colon H \to H \tensor H} \,.\, k_i \mbox{ copyable
        along } \delta \}. 
  \end{align*}
\end{proposition}
\begin{proof}
  Since every $C \in \cC(A)$ is finite-dimensional and hence
  $\Proj(C)$ is atomic, an orthocomplemented
  sublattice $L$ as in Theorem~\ref{thm:CA} is completely determined by its
  atoms $(k_i)$.
  \qed
\end{proof}

Every finite-dimensional C*-algebra is a von Neumann algebra, and
hence in case $H$ is finite-dimensional and $A=\Cat{Hilb}(H,H)$, we
find that $\cC(A)$ is the collection of all commutative C*-subalgebras
of $A$. Notice that the characterization of Proposition~\ref{prop:fdCAker} above
has no need for the cumbersome combinatorial symmetry considerations of
\cite[1.4.5]{heunenlandsmanspitters:bohrification}. 
Corollary~\ref{cor:characterizationcopyable} gives another characterization of
$\cC(A)$ in the finite-dimensional case, purely in terms of classical
structures and their morphisms.

\begin{proposition}
\label{prop:fdCA}
  For a finite-dimensional Hilbert space $H$ and the von Neumann
  algebra $A=\Cat{fdHilb}(H,H)$:
  \begin{align*}
    \cC(A) \cong \{ (\delta_i)_{i \in I} \mid \; & 
      \delta_i \in \Cat{CS}[\Cat{fdHilb}], \\
      & \delta_i, \delta_j \mbox{ partially
        complementary for } i \neq j, \\
      & \exists_{\delta \in \Cat{CS}[\Cat{fdHilb}]}\forall_i \,.\,
      \Cat{CS}[\Cat{fdHilb}](\delta_i,\delta) \neq \emptyset \}.
  \end{align*}
  Hence $\cC(A)$ is isomorphic to the collection of cocones in the
  category of classical substructures on $H$ that are pairwise
  partially complementary. 
  \qed
\end{proposition}

Another, very concrete, characterization in terms of traces and
determinants is
known in the finite-dimensional
case~\cite[Proposition~1.3]{parthasarathy:estimation}. It is easy to
compute, but the above characterization seems conceptually more 
informative and lends itself more readily to generalization.

\section{Concluding remarks}
\label{sec:conclusion}

Observing the similarities across the three levels of quantum mechanics
considered, we can now propose the following precise formulation of
complete complementarity.
\begin{quote}
  A collection of classical structures is \emph{completely
  complementary} when its members are pairwise partially complementary and
  jointly epic. 
\end{quote}
Notice that this formulation is almost information-theoretic.
Compare also \cite{vandenbergheunen:commutativity}; this suggests that
a dagger kernel category could be seen as a colimit (or
amalgamation of some other kind) of its Boolean subcategories.

The view on $\cC(A)$ provided by Section~\ref{sec:vonneumannalgebras} holds
several promises for the study of functors on $\cC(A)$ that we intend
to explore further in future work: 
\begin{itemize}
  \item One can consider variations in the study of $\Cat{Set}$-valued
    functors on $\cC(A)$ by choosing different morphisms on $\cC(A)$: \eg
    inclusions~\cite{heunenlandsmanspitters:bohrification}, or reverse 
    inclusions~\cite{doeringisham:thing}. In the above perspective,
    the natural direction that suggests itself is that of morphisms
    between classical structures, \ie inclusions. Moreover, a more 
    interesting choice of morphisms based on classical structures (see
    \eg\cite{coeckepavlovic:classicalobjects}) could make $\cC(A)$
    into a category that is not just a partially ordered set.  
  \item The topos of functors on $\cC(A)$ can be abstracted away from
    $\Cat{Hilb}$ to any dagger monoidal kernel category that
    satisfies a suitable `spectral assumption' linking commutative
    submonoids of endohomsets to classical structures. For example, one
    could lift Theorem~\ref{thm:CA} or even Proposition~\ref{prop:fdCA} to a definition,
    and study $\Cat{Set}$-valued functors on these characterizations of
    $\cC(A)$ in any dagger monoidal kernel category.
   
    In fact, in this generalized setting, there is no need for the
    base category to be $\Cat{Set}$. After all, the basic objects of
    study of \eg\cite{doeringisham:thing} are really partial orders of
    subobjects in a functor category. This just happens to be a
    Heyting algebra because the functors take values in the topos
    $\Cat{Set}$, but in principle less structured partial orders of
    subobjects are just as interesting, and perhaps are also
    justifiable physically. 
  \item One of the weak points of the study of functors on $\cC(A)$ to
    date is that there is no obvious way to study compound
    systems. That is, there is no obvious (satisfactory) relation
    between $\cC(A \tensor B)$ and $\cC(A)$ and $\cC(B)$. Considering
    $A$ as (a submonoid of) an endohomset opens the broader context of
    a fibred setting in which studying entanglement is possible. 
\end{itemize}
All in all, the above considerations strongly suggest studying
fibrations of all classical structures over all objects of a dagger
(kernel) monoidal category, \ie studying the forgetful functor
$\Cat{CS}[\cat{D}] \to \cat{D}$. 

Finally, we remark that we have not used the H*-axiom (or the
Frobenius equation) at all in this paper. Apparently, the combination
of (copyable) kernels with the dagger monic type $X \to X \tensor X$
of classical structures suffices for these purposes.

{\small{
  \bibliographystyle{plain}
  \bibliography{copyables-fp}
}}

\end{document}